\documentclass[onecolumn]{IEEEtran}

\usepackage{graphicx}
\usepackage{amsmath,amssymb}
\usepackage{cases}
\usepackage{amsfonts}
\usepackage{indentfirst}
\usepackage{slashbox}
\usepackage{cite}
\usepackage{caption}
\usepackage{algorithm}
\usepackage{algpseudocode}
\usepackage{booktabs}
\usepackage{subfigure}
\usepackage{multirow}
\usepackage{amsthm}

\theoremstyle{remark}
\newtheorem{remark}{Remark}
\theoremstyle{}
\newtheorem{theorem}{Theorem}
\theoremstyle{}

\newtheorem{lemma}{Lemma}
\theoremstyle{}
\newtheorem{definition}{Definition}
\theoremstyle{remark}
\newtheorem{example}{Example}
\theoremstyle{definition}

\makeatletter
\newcommand{\tabcaption}{\def\@captype{table}\caption}
\makeatletter
\ifCLASSINFOpdf

\else

\fi

\author{Qifa~Yan,~
        Minquan~Cheng,~
        Xiaohu~Tang,~
        and~Qingchun~Chen
\thanks{Q. Yan, X. Tang, and Q. Chen are with The School of Information Science and Technology, Southwest Jiaotong University, Chengdu, China, 610031. E-mails: qifa@my.swjtu.edu.cn,  xhutang@swjtu.edu.cn, qcchen@swjtu.edu.cn.}
\thanks{M. Cheng is with Guangxi Key Lab of Multi-source Information Mining \& Security, Guangxi Normal University, Guilin 541004, China.
 E-mail: chengqinshi@hotmail.com.}}
\begin{document}

\title{On the Placement Delivery Array Design in Centralized Coded Caching Scheme}



\maketitle

\begin{abstract}
Caching is a promising solution to satisfy the ever increasing demands for the multi-media traffics. In caching networks, coded caching is a recently proposed technique that achieves significant performance gains over the uncoded caching schemes. However, to implement the coded caching schemes, each file has to be split into $F$ packets, which usually increases exponentially with  the number of users $K$.  Thus, designing caching schemes that decrease the order of $F$ is meaningful for practical implementations. In this paper, by reviewing the Ali-Niesen caching scheme, the placement delivery array (PDA) design problem is firstly formulated to characterize the placement issue and the delivery issue with a single array. Moreover, we show  that, through designing appropriate PDA, new centralized coded caching schemes can be discovered. Secondly, it is shown that the Ali-Niesen scheme corresponds to a special class of PDA, which realizes the best coding gain with the least $F$. Thirdly, we present a new construction of PDA for the centralized caching system, wherein the cache size of each user $M$ (identical cache size is assumed at all users) and the number of files $N$ satisfies   $M/N=1/q$ or ${(q-1)}/{q}$ ($q$ is an integer such that $q\geq 2$). The new construction can decrease the required $F$ from the order $O\left(e^{K\cdot\left(\frac{M}{N}\ln \frac{N}{M} +(1-\frac{M}{N})\ln \frac{N}{N-M}\right)}\right)$ of Ali-Niesen scheme to $O\left(e^{K\cdot\frac{M}{N}\ln \frac{N}{M}}\right)$ or $O\left(e^{K\cdot(1-\frac{M}{N})\ln\frac{N}{N-M}}\right)$ respectively, 
while the coding gain loss is only $1$.
\end{abstract}

\begin{IEEEkeywords}
 Placement Delivery Array, Centralized Coded Caching,  Content Delivery Network.
\end{IEEEkeywords}
\IEEEpeerreviewmaketitle
\section{Introduction}
 Driven by the broadband services such as video-on-demand (VoD) and catch-up TV, wireless data traffic is predicted to increase dramatically in the next few years, up to two orders of magnitude by 2020\cite{Cisco}. However, it is widely acknowledged that the current wireless architecture can not effectively support the ever increasing traffic, in particular when the allocated spectrum resource is limited, which consequentially leads to congestion in peak times. Fortunately, there is a common feature of asynchronous content reuse in video streaming applications, $i.e.$, the same content is requested by different users at different times \cite{caire2013femtocaching}. As a result, one promising approach to alleviate the congestion is ``removing" some traffic delivery by disseminating contents into memories across the network when the network load is low. This dissemination is referred to as \emph{content placement} or \emph{content caching}. During the content delivery at peak times, different users' requests can be better satisfied with the help of content caching. Therefore, the congestion can be effectively alleviated through two separate phases, namely, the \emph{content placement phase} at off peak times and the \emph{content delivery phase} at peak times, respectively.

Previous researches  focused on either content placement or content delivery design. For example, different content placements and content cache location choices were studied in \cite{Meyerson2001,Borst2010} to improve the quality-of-service. The allocations of the replicated contents in the caches across the network were investigated in \cite{Baev2008,Korupolu1999} to reduce the average access cost, while the multicast gains to users with the similar demands were addressed in \cite{Almeroth1996,Dan1996}. The key idea of these conventional caching design is to deliver part of the content to caches close to the end users, so that the requested content can be served locally. 
Recently in the seminal work of coded caching scheme in \cite{maddah2013fundamental}, Maddah-Ali and Niesen showed that
better coding gain can be achieved by exploiting caches to create multicast opportunities.
Indeed, the system studied in \cite{maddah2013fundamental} is a centralized caching scheme, where a central server coordinates all the transmissions.
By jointly designing the content placement phase and the content delivery phase, the central sever is able to simultaneously deliver distinct contents to different users through a shared link.
In order to simplify the design, Ali and Niesen adopted the following two strategies:
 \begin{itemize}
   \item [S$1$.] In the content placement phase, identical caching policy for different files is assumed for each user, 
   $i.e.$, each user  caches packets with the same indices from all files, where packets belonging to every file is ordered according to a chosen numbering;
   \item [S$2$.] In the content delivery phase, the requested packets by users will be XOR multiplexing to formulate the delivered signal.
\end{itemize}

It is shown that, with an elaborate caching design in content placement phase and the XOR multiplexing of those requested packets in content delivery phase, the near-optimal scheme can be realized for arbitrary traffic demands by using the Ali-Niesen scheme in \cite{maddah2013fundamental}.
Nevertheless, in general, one disadvantage of Ali-Niesen scheme is that, it usually needs to split each file into $F$ packets, where $F$ increases exponentially with the number of users $K$. This would become infeasible when $K$ is  large. Therefore, designing a coded caching scheme with smaller size $F$ will be a critical issue, especially for practical implementations.

In this paper, just like the Ali-Niesen scheme, we assume the same two strategies in the coded caching scheme design as well. However, unlike the previous analysis, firstly, we propose to use \emph{placement delivery array (PDA)} to characterize the involved strategies in two phases.
Essentially, PDA is an $F\times K$ array consisting of a specific symbol $``*"$ and some integers, wherein the row index $j$ ($j\in \{0,1,\cdots,F-1\}$) and the column index $k$ ($k\in \{0,1,\cdots,K-1\}$) stands for the $j$-th packets of all the files for the $k$-th user. In particular, a PDA can depict the following coded caching scheme
\begin{itemize}
\item[$1.$] In the content placement phase, the symbol $``*"$  at row $j$ and column $k$ in PDA indicates that user $k$ stores the $j$-th packet of all the files;
\item[$2.$] In the content delivery phase, the request packets by different users, which are indicated by the same integers at each row, will be sent by the sever simultaneously after XORing operation.
\end{itemize}
In this way, we can depict the placement and delivery schemes for all possible requests in a single array. As a result, an appropriate PDA design would recover new centralized coded caching schemes. In fact, Ali-Niesen scheme corresponds to the so-called regular PDA, where the occurrence of each integer
is a constant. Notably, we establish an upper bound on the achievable coding gain of the regular PDA, which reveals that Ali-Niesen scheme achieves the upper bound with the least $F$. Unfortunately, $F$ increases exponentially with the number of users $K$. Thus, designing caching schemes that decrease the order of $F$ is meaningful for practical implementations. In this paper, we aim at the centralized coded caching scheme design with reduced requirement in the size of $F$. More specifically, for a system with $N$ files, and each user can cache $M$ equivalent files, such that $M/N=1/q$ or ${(q-1)}/{q}$, we present  new constructions of PDA that significantly decrease the required $F$ from the order $O\left(e^{K\cdot\left(\frac{M}{N}\ln \frac{N}{M}+(1-\frac{M}{N})\ln \frac{N}{N-M}\right)}\right)$ of Ali-Niesen scheme to $O\left(e^{K\cdot\frac{M}{N}\ln \frac{N}{M}}\right)$ ($M/N=1/q$) and $O\left(e^{K\cdot(1-\frac{M}{N})\ln\frac{N}{N-M}}\right)$ ($M/N={(q-1)}/{q}$) respectively,  both saving a factor in $F$ that increases exponentially with $K$, while the coding gain loss is only $1$.

The remainder of this paper is organized as follows. In Section \ref{sec_model}, the system model and Ali-Niesen scheme are briefly reviewed. In Section \ref{sec_PDmat},
the definition of PDA is introduced and then  its connection with coded caching scheme is established. In Section \ref{sec_PDA}, the Ali-Niesen scheme is re-explained by using the regular PDA and its optimality is proved in terms of its capability to approach the upper bound on the coding gain for regular PDA. In Sections \ref{sec_nPDA} and  \ref{sec_comp}, a new PDA construction and its performance comparison with existing schemes are presented, respectively. Finally, the conclusion is given in Section \ref{sec_conclusion}.

\textbf{Notations}: In this paper, arrays are denoted by bold capital letters, vectors are denoted by bold lower case letters.  We use $[i,j]$ to denote the set of integers $\{i,i+1,\cdots,j\}$ and $[i,j)$  to denote the set $\{i,i+1,\cdots,j-1\}$.  For two series $\{a_n\}, \{b_n\}$, $a_n\sim b_n$ means $\lim_{n\rightarrow\infty}\frac{a_n}{b_n}=1$. The operation $\oplus$ means bitwise Exclusive OR (XOR) operation of two packets. The set of positive integers is denoted by $\mathbb{N}^{+}$.

\section{Network Model And Ali-Niesen Scheme}\label{sec_model}

Let us consider a caching system consisting of one server, connected by $K$ users through an error-free shared link. The server has $N$ files ($N\geq K$), which are denoted by $\mathcal{W}=\{W_0,W_1,\cdots,W_{N-1}\}$. Without loss of generality (W.L.O.G.), we assume that each file is of unit length.  Denote the $K$ users by $\mathcal{K}=\{0,1,\cdots,K-1\}$,  each having a cache  of size $M$ units, where $0\leq M\leq N$. Fig. \ref{system} shows a diagram of the aforementioned caching system.

\begin{figure}[htbp]
\centering\includegraphics[width=0.5\textwidth]{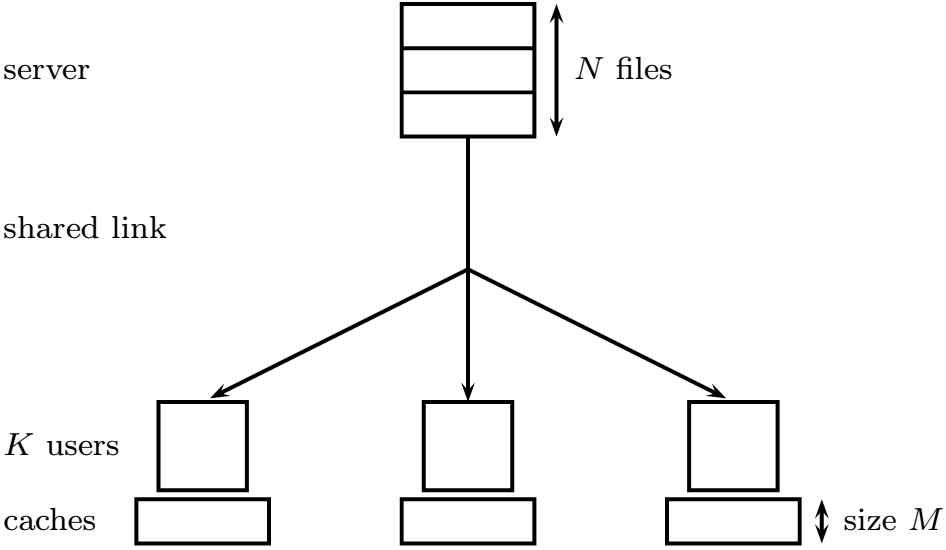}
\caption{Caching system}\label{system}
\end{figure}

The caching system operates in two separated  phases:
\begin{enumerate}
\item [1.] In the placement phase,  a file is sub-divided into $F$ equal packets\footnote{ Memory sharing technique may lead to non equally divided packets\cite{maddah2013fundamental}, in this paper, we will not discuss this case.}, \emph{i.e.}, $W_{i}=\{W_{i,j}:j\in[0,F)\}$, each of size $1/F$ units. Then, each packet is placed in different user caches deterministically. Denote the contents at user $k$ by $\mathcal{C}_{k}$, where $k\in\mathcal{K}$. The total size of packets at each user should not exceed its cache size $M$, \emph{i.e.}, the number of packets in $\mathcal{C}_k$ is at most $\lfloor MF\rfloor$.
\item [2.] In the delivery phase, each user randomly requests one file from the files set $\mathcal{W}$ independently. The request is denoted by $\mathbf{d}=(d_0,d_1,\cdots,d_{K-1})$, which indicates that user $k$ requests the $d_k$-th file $W_{d_k}$ for any $d_k\in[0,N)$ and $k\in\mathcal{K}$. Once the server received the users' request $\mathbf{d}$, it broadcasts a signal of at most $S_{\mathbf{d}}$ packets to users. Each user is able to recover its requested file from the signal received in the delivery phase
with the help of the contents within its own cache.
\end{enumerate}

The caching system is parameterized by $K,M,N$, so in this paper we call it a $(K,M,N)$ caching system. Following the convention, we refer to a realization of placement  and  delivery  as  a caching scheme. In a caching scheme, if each file is sub-divided into $F$ packets, we refer to such a  scheme as an $F$-division caching scheme. Specifically, we define the  rate of the $F$-division scheme as
\begin{align}
R=\sup_{\tiny{\mbox{$\begin{array}{c}
                 \mathbf{d}=(d_0,\cdots,d_{K-1}) \\
                 d_k\in[0,N),\forall k\in[0,K)
               \end{array}
$}}} \left\{\frac{S_{\mathbf{d}}}{F}\right\}\label{eq_def_R}.
\end{align}
The definition can be explained as follows: Assume that in the placement phase, the server constructs a code for each possible request. During the delivery phase,  the server responds the users' request $\mathbf{d}$ by sending the corresponding code of length $S_{\mathbf{d}}/F$. In the sense of source coding, $R$ is the worst coding rate (normalized by the file size) over all possible requests $\mathbf{d}$ \cite{ji2015}. 
That is, the smaller $R$, the less load at the server. Therefore, the primary concern for a given  $(K,M,N)$ caching system is to minimize the rate $R$ .

In \cite{maddah2013fundamental}, Ali and Niesen  proposed a caching scheme for the $(K,M,N)$ caching system under the assumption that the files can be arbitrarily divided. Algorithm \ref{alg0} depicts  a ${K\choose KM/N}$-division Ali-Niesen scheme where $M/N\in\{0, 1/K,2/K,\cdots,1\}$.
It was  shown in \cite{maddah2013fundamental} that, the scheme  is feasible and each user can successfully recover its requested file at a rate
\begin{align}\label{Eqn_Ali_Rate}
R_{A-N}\left(\frac{M}{N},K\right)=K\left(1-\frac{M}{N}\right)\cdot\frac{1}{1+\frac{KM}{N}}
\end{align}
for $M/N\in\{0, 1/K,2/K,\cdots,1\}$. For general $0\leq M/N\leq 1$, the lower convex envelop of these points can be achieved by memory sharing technique \cite{maddah2013fundamental}.

\begin{algorithm}[htb]
\caption{Ali-Niesen Caching Scheme}\label{alg0}
\begin{algorithmic}[1]
\Procedure {Placement}{$W_0,W_1,\cdots,W_{N-1}$}
\State $t\leftarrow\frac{KM}{N}$
\State $\mathfrak{T}\leftarrow\{\mathcal{T}\subset \mathcal{K}:|\mathcal{T}|=t\}$
\For{$n\in[0,N)$}
\State Split $W_n$ into $\{W_{n,\mathcal{T}}:\mathcal{T}\in\mathfrak{T}\}$ of equal packets
\EndFor
\For{$k\in\mathcal{K}$}
\State $\mathcal{C}_k\leftarrow\{W_{n,\mathcal{T}}:n\in[0,N),\mathcal{T}\in\mathfrak{T},k\in\mathcal{T}\}$
\EndFor
\EndProcedure
\Procedure{Delivery}{$W_0,W_1,\cdots,W_{N-1},d_0,d_1,\cdots,d_{K-1}$}
\State  $t\leftarrow\frac{KM}{N}$
\State $\mathfrak{S}\leftarrow\{\mathcal{S}\subset\mathcal{K}:|\mathcal{S}|=t+1\}$
\State Server sends $\{\oplus_{k\in\mathcal{S}} W_{d_k,\mathcal{S}\backslash \{k\}}:\mathcal{S}\in\mathfrak{S}\}$
\EndProcedure
\end{algorithmic}
\end{algorithm}

For clarity, we trim an example from \cite{maddah2013fundamental} to illustrate the scheme.

\begin{example}\label{exam1}
Consider the case $N=K=2$, $i.e.$, there are two files, say $W_0=A$, $W_1=B$ and two users each with cache memory size $M=1$.
According to Algorithm 1, we have $t=KM/N=1$. Then $A,B$ are partitioned into ${K\choose t}={2\choose 1}=2$ packets of equal size, $i.e.$ $A=\{A_{\{0\}},A_{\{1\}}\}$ and $B=\{B_{\{0\}},B_{\{1\}}\}$. By simplicity, we abbreviate $\{0\}$ and  $\{1\}$ as $0$ and $1$, respectively. That is, $A=\{A_0,A_1\}$ and $B=\{B_0,B_1\}$. In the placement phase, user 0 caches $\mathcal{C}_0=\{A_0,B_0\}$, while user 1 caches $\mathcal{C}_1=\{A_1,B_1\}$. In the delivery phase, assume that user 0 requests $A$ and user 1 requests $B$. Since user 0 has packet $A_0$ of $A$, it only needs to obtain the missing packet $A_1$, which is cached by user 1. Similarly, user 1 needs only to obtain the packet $B_0$ which is in the cache of user 0. Thus, the server can simply send $A_1\oplus B_0$. Since user 0 has $B_0$ in its cache, it can decode $A_1$ from the signal $A_1\oplus B_0$. Similarly, user 1 can decode its missing packet $B_0$. The signals sent for other three requests are illustrated in Fig. \ref{multicast}. It is easy to see that this scheme achieves a rate $1/2$ since each packet $A_0,A_1, B_0, B_1$ is of length
$1/2$ unit, so does their XOR values.

\begin{figure}[htbp]
\centering\includegraphics[width=0.5\textwidth,height=0.5\textwidth]{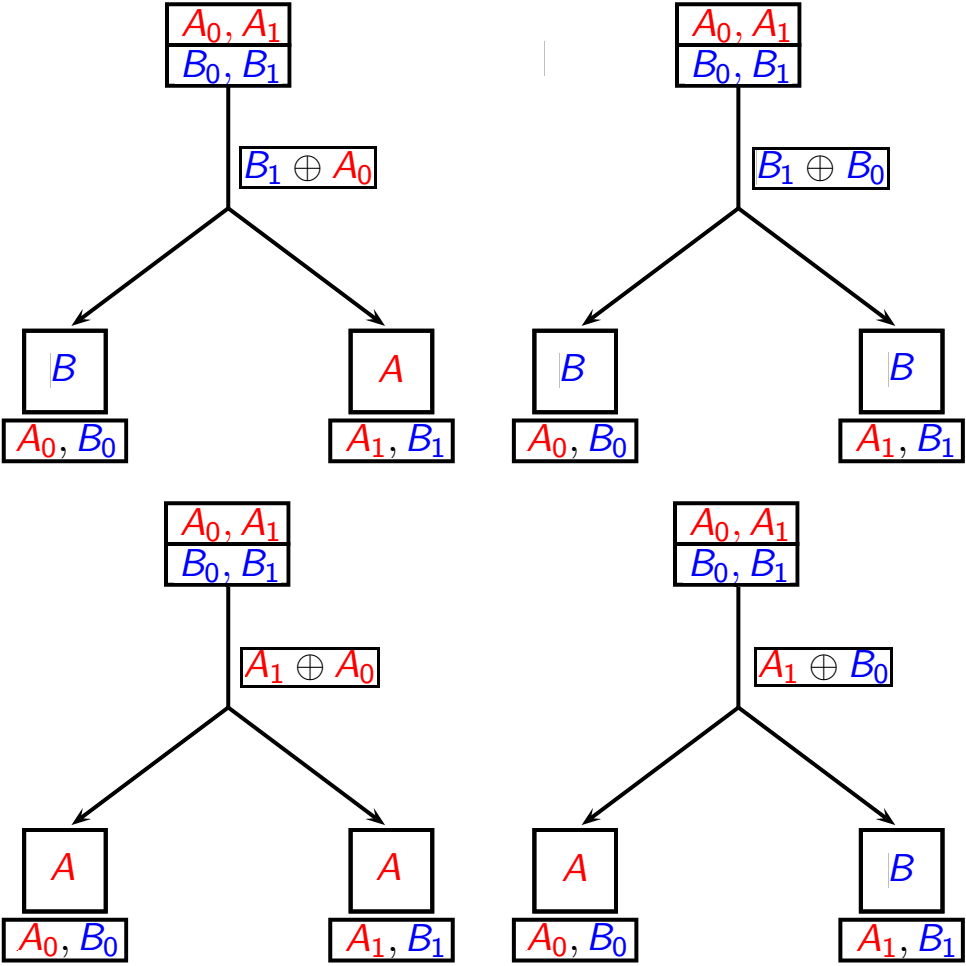}
\caption{Caching scheme for $N=K=2,M=1$ with all four possible user requests.
}\label{multicast}
\end{figure}
\end{example}

Nevertheless, for the same $(K,M,N)$ caching system, the conventional non-coded scheme has rate
\begin{align}\label{Eqn_Conv_Rate}
R_{\mathrm{Conventional}}\left(\frac{M}{N},K\right)=K\left(1-\frac{M}{N}\right)
\end{align}
Compare \eqref{Eqn_Ali_Rate} with \eqref{Eqn_Conv_Rate}, it is seen that the gain for conventional
non-coded scheme is relevant to the the ratio of local
cache memory $M$ to the total content $N$, whereas a new  coding gain  for Ali-Niesen scheme
comes from aggregate global cache size $KM$ through XOR coding, even though there is no cooperation among
the users.

\section{Placement Delivery Array}\label{sec_PDmat}

In this section, we propose a new concept  to characterize the caching system.

\begin{definition}
For  positive integers $K,F, Z$ and $S$, an $F\times K$ array  $\mathbf{P}=[p_{j,k}]$, $j\in [0,F), k\in[0,K)$, composed of a specific symbol $``*"$  and $S$ nonnegative integers
$0,1,\cdots, S-1$, is called a $(K,F,Z,S)$ placement delivery array (PDA) if it satisfies the following conditions:
\begin{enumerate}
  \item [C$1$.] The symbol $``*"$ appears $Z$ times in each column;
  \item [C$2$.] Each integer occurs at least once in the array;
  \item [C$3$.] For any two distinct entries $p_{j_1,k_1}$ and $p_{j_2,k_2}$,    $p_{j_1,k_1}=p_{j_2,k_2}=s$ is an integer  only if
  \begin{enumerate}
     \item [a.] $j_1\ne j_2$, $k_1\ne k_2$, i.e., they lie in distinct rows and distinct columns; and
     \item [b.] $p_{j_1,k_2}=p_{j_2,k_1}=*$, i.e., the corresponding $2\times 2$  subarray formed by rows $j_1,j_2$ and columns $k_1,k_2$ must be of the following form
  \begin{eqnarray*}
    \left[\begin{array}{cc}
      s & *\\
      * & s
    \end{array}\right]~\textrm{or}~
    \left[\begin{array}{cc}
      * & s\\
      s & *
    \end{array}\right]
  \end{eqnarray*}
   \end{enumerate}

\end{enumerate}
\end{definition}

Based on a $(K,F,Z,S)$  PDA $\mathbf{P}=[p_{j,k}]$ with $j\in [0,F)$ and $k\in[0,K)$,  an $F$-division  caching scheme for a $(K,M,N)$ caching system with $M/N=Z/F$ can be  obtained as follows:
\begin{itemize}
\item [1.] \textbf{Placement Phase:} All the files are cached in the same manner. Each file $W_i$ is split into $F$ packets, $i.e.$ $W_i=\{W_{i,j}:j\in[0,F)\}$, $\forall~i\in[0,N)$,
so that user $ k\in\mathcal{K}$ caches  packets
      \begin{align}
      \mathcal{C}_k=\{W_{i,j}: p_{j,k}=*,\forall~i\in[0,N)\}\label{eq_p1}
      \end{align}
     By C$1$, each user stores $N\cdot Z$ packets. Since each packet has size $1/F$, the whole size of cache is $N\cdot Z\cdot 1/F=M$, which satisfies the users' cache constraint.
\item [2.] \textbf{Delivery Phase:}  Once the server receives the request $\mathbf{d}=(d_0,d_1,\cdots,d_{K-1})$, at the  time slot $s$, $0\le s<S$, it sends
      \begin{align}
      \bigoplus_{p_{j,k}=s,  j\in[0, F),  k\in[0,K)}W_{d_{k},j}\label{eq_p2}
      \end{align}
\end{itemize}

Assume that in PDA $\mathbf{P}$ there are $l$ entries $p_{j_1,k_1}=p_{j_2,k_2}=\cdots=p_{j_l,k_l}=s$ where $0\le j_1,\cdots,j_l< F$ and
$0\le k_1,\cdots,k_l<K$.  Consider the subarray  formed by rows $j_1,\cdots,j_l$ and columns $k_1,\cdots,k_l$, which is of order  $l\times l$ since $j_{u}\ne j_{v}$ and $k_u\ne k_v$ for all $1\le u\ne v\le l$ by C$3$-a. Further, applying C$3$-b we have $p_{j_u,k_v}=*$ for all $1\le u\ne v\le l$. That is to say,
this subarray is equivalent to the following $l\times l$ array
\begin{eqnarray}\label{Eqn_Matrix_1}
    \left[\begin{array}{cccc}
      s & * & \cdots & *\\
      * & s & \cdots & *\\
      \vdots & \vdots &\ddots & \vdots\\
      * & * & \cdots & s
    \end{array}\right]
\end{eqnarray}
with respect to row/column permutation. According to \eqref{eq_p2}, at the  time slot $s$, $0\le s<S$, the sever sends
\begin{eqnarray*}
      \bigoplus_{1\le u\le l}W_{d_{k_u},j_u}
\end{eqnarray*}
Note from \eqref{Eqn_Matrix_1} that in column $v$, all the entries are $``*"$s except for the $v$th one. Then
it follows \eqref{eq_p1} that user $k_v$ has already all the other packets $W_{d_{k_u},j_u}$, $1\le u\ne v\le l$, in its cache at the placement phase. Then,
it can easily decode the desired packet $W_{d_{k_v},j_v}$. Since the server sends $S$ packets for each possible request $\mathbf{d}$, the rate of the scheme is given by $S/F$.

By the above analysis, we have proved the following theorem.

\begin{theorem}\label{thm1} For a given $(K,F,Z,S)$ PDA $\mathbf{P}=[p_{j,k}]_{F\times K}$, there exists a corresponding $F$-division caching scheme for any $(K,M,N)$ caching system with $M/N=Z/F$. Precisely, each user is able to decode its requested file correctly for any request $\mathbf{d}$ at the rate
$R=S/F$.
\end{theorem}

According to Theorem \ref{thm1}, we can find an $F$-division caching scheme according to a PDA of size $F\times K$.   This suggests, through exploring constructions of PDA, it is possible to discover new caching schemes. Actually, PDA can be used to describe a subclass of caching schemes employing strategies S$1$ and S$2$, which includes Ali-Niesen scheme. For instance, the scheme in Example \ref{exam1} can be described by a $ (2,2,1,1)$ PDA
\begin{align}
\left[\begin{array}{cc}
        * & 0 \\
        0 & *
      \end{array}
\right]\notag
\end{align}
In the next section, we will rebuild Ali-Niesen scheme by means of PDA  in detail.

Finally, we conclude this section by an illustrative example.

\begin{example}\label{exam2} It is easily checked that  the array
\begin{align}
\mathbf{A}^{2,2} =\left[\begin{array}{cccccc}
 *&  1&  *&  2&  *&  0\\
 0&  *&  *&  3&  1&  *\\
 *&  3&  0&  *&  2&  *\\
 2&  *&  1&  *&  *&  3
\end{array}
\right]\label{eq_A}
\end{align}
is a $(6,4,2,4)$   PDA.

Assume that in a $(6,3,6)$ caching system, the files are denoted by $W_0,W_1,W_2,W_3,W_4,W_5$ respectively, and each file is divided into $F=4$ packets, $i.e.$ $W_i=\{W_{i,0},W_{i,1},W_{i,2},W_{i,3}\}$ , $i\in [0,6)$. This system can be implemented according to $\mathbf{A}^{2,2} $ as follows:
\begin{itemize}
   \item \textbf{Placement Phase}:   The contents in each users are
       \begin{align*}
       \mathcal{C}_0=\left\{W_{i,0},W_{i,2}:i\in[0,6)\right\}~~~~ \mathcal{C}_1=\left\{W_{i,1},W_{i,3}:i\in[0,6)\right\}\notag\\
       \mathcal{C}_2=\left\{W_{i,0},W_{i,1}:i\in[0,6)\right\}~~~~ \mathcal{C}_3=\left\{W_{i,2},W_{i,3}:i\in[0,6)\right\}\notag\\
       \mathcal{C}_4=\left\{W_{i,0},W_{i,3}:i\in[0,6)\right\}~~~~ \mathcal{C}_5=\left\{W_{i,1},W_{i,2}:i\in[0,6)\right\}\notag
       \end{align*}
   \item \textbf{Delivery Phase}: Assume the request vector is $\mathbf{d}=(0,1,\cdots,5)$. Table \ref{table1} shows the transmitting process.
   \begin{table}[!htp]
  \centering
  \caption{Deliver steps in Example \ref{exam2} }\label{table1}
  \normalsize{
  \begin{tabular}{cc}
    \toprule
   Time Slot& Transmitted Signnal  \\
   \midrule
   $0$&$W_{0,1}\oplus W_{2,2}\oplus W_{5,0}$\\
   $1$&$W_{1,0}\oplus W_{2,3}\oplus W_{4,1}$\\
  $2$& $W_{0,3}\oplus W_{3,0}\oplus W_{4,2}$\\
   $3$& $W_{1,2}\oplus W_{3,1}\oplus W_{5,3}$\\
 \bottomrule
  \end{tabular}}
\end{table}
\end{itemize}

\end{example}

\section{PDA For Ali-Niesen Scheme And Its Optimality}\label{sec_PDA}

In this section, we firstly  prove that Ali-Niesen scheme corresponds to a special class of PDA, regular PDA. Next, we
show its optimality by establishing an upper bound on the coding gain for the regular PDA.

Recall that, in Ali-Niesen scheme, let $t=KM/N\in[0,K]$, then a file is split into ${K\choose t}$ packets and each packet is labeled by a subset
of size $t$ in the set $\mathcal{K}=\{0,1,\cdots,K-1\}$. In the delivery phase, Algorithm \ref{alg0} visits each subset  of size $t+1$ one by one. Let us arrange all  such subsets in the lexicographic order
and then define $f_{t+1}(\mathcal{S})$ to be its order minus 1 for any subset $\mathcal{S}$ of size $t+1$. Clearly,
$f_{t+1}$ is a bijection from $\{\mathcal{S}\subset\mathcal{K}:|\mathcal{S}|=t+1\}$ to $[0,{K\choose t+1})$.
For example, when $K=4$ and $t=1$, in $\mathcal{K}=\{0,1,2,3\}$ all the subsets of size $t+1=2$ are ordered as
$\{0,1\},\{0,2\},\{0,3\}, \{1,2\}, \{1,3\}$, and $\{2,3\}$. Accordingly,
$f_2(\{0,1\})=0,f_2(\{0,2\})=1,f_2(\{0,3\})=2, f_2(\{1,2\})=3, f_2(\{1,3\})=4$, and $f_2(\{2,3\})=5$.

Let $\mathbf{D}^{K,t}$ be a ${K\choose t}\times K$ array. Denote its rows by the sets in
$\{\mathcal{T}\subset \mathcal{K}:|\mathcal{T}|=t\}$, and columns by $0,1,\cdots, K-1$, respectively. Then, define the
entry $d_{\mathcal{T},k}$ in row $\mathcal{T}$ and column $k$ as
\begin{align}
d_{\mathcal{T},k}=\left\{\begin{array}{cc}
                                         *, & \mbox{if}~k\in\mathcal{T} \\
                                         f_{t+1}(\mathcal{T}\cup\{k\}), & \mbox{if}~k\notin\mathcal{T}
                                       \end{array}
\right.\label{eq_Akt}
\end{align}
It is easily seen from Algorithm 1 that both the placement and delivery of Ali-Niesen scheme are essentially applications of
$\mathbf{D}^{K,t}$ to \eqref{eq_p1}.  Indeed, $\mathbf{D}^{K,t}$ is a regular PDA.

\begin{definition}
An array $\mathbf{P}$ is said to be a $g$-regular $(K,F,Z,S)$  PDA, $g$-$(K,F,Z,S)$ PDA or $g$-PDA for short, if it satisfies
\rm{C$1$}, \rm{C$3$}, \emph{and
\begin{itemize}
\item [C$2'$.] Each integer appears $g$ times in $\mathbf{P}$ where $g$ is a constant.
\end{itemize}}
\end{definition}

In the corresponding caching scheme, a regular PDA leads to a fact that each packet sent is intended to $g$ users. In what follows, we refer  to $g$ as \textit{coding gain} for a $g$-$(K,F,Z,S)$~PDA and its corresponding caching scheme. Obviously, the coding gain $g$ is desirable to be as large as possible.
The following theorem shows that the Ali-Niesen scheme can  be determined by a regular PDA with gain $g=t+1$.

\begin{theorem}\label{lem7} For a $(K,M,N)$ caching system, such that  $M/N\in \{0,1/K,2/K,\cdots,1\}$, let $t=KM/N$, then the  array $\mathbf{D}^{K,t}$ in \eqref{eq_Akt} which corresponds to Ali-Niesen scheme, is a $(t+1)$-$(K,F,Z,S)$ PDA, where $F={K\choose t}$, $Z={K-1\choose t-1}$, and $S={K\choose t+1}$.
\end{theorem}
\begin{proof}
From \eqref{eq_Akt}, $\mathbf{D}^{K,t}$  is a ${K\choose t}\times K$ array consisting of $``*"$ and integers in $[0,{K\choose t+1})$.
Hereafter, it suffices to check C$1$, C$2'$, and C$3$.

For each $k\in\mathcal{K}$, $k$ is included in exactly ${K-1\choose t-1}$ subsets of $\mathcal{K}$ in $\{\mathcal{T}\subset \mathcal{K}:|\mathcal{T}|=t\}$. Thus by \eqref{eq_Akt},
the symbol $``*"$ appears $Z={K-1\choose t-1}$ times in each column exactly. That is, C$1$ is satisfied.

Next, assume that two distinct entries $d_{\mathcal{T}_1,k_1}=d_{\mathcal{T}_2,k_2}=s$ where $\mathcal{T}_1,\mathcal{T}_2\subset \mathcal{K}$ with
$|\mathcal{T}_1|=|\mathcal{T}_2|=t$ and $k_1,k_2\in \mathcal{K}$. Then applying the fact that
$f_{t+1}$  is a bijection from $\{\mathcal{S}\subset\mathcal{K}:|\mathcal{S}|=t+1\}$ to $[0,{K\choose t+1})$,  from \eqref{eq_Akt} we have that
$s$ is an integer if and only if
\begin{align*}
\mathcal{T}_1\cup \{k_1\}=\mathcal{T}_2\cup \{k_2\}
\end{align*}
which implies that
\begin{itemize}
\item Each integer $f_{t+1}(\mathcal{S})$ ($|\mathcal{S}|=t+1$) in $[0,{K\choose t+1})$ occurs exactly $t+1$ times since there are $t+1$ distinct
possibilities of $(\{k_1\}, \mathcal{T}_1=\mathcal{S}\setminus\{k_1\})$ with $k_1$ ranging over $\mathcal{S}$;
\item $\mathcal{T}_1\ne \mathcal{T}_2$ and $k_1\ne k_2$, $i.e.$ the two entries are in distinct rows and columns.  Further, this is equivalent to
 $k_1\in \mathcal{T}_2, k_2\in \mathcal{T}_1$, and thus $d_{\mathcal{T}_1,k_2}=d_{\mathcal{T}_2,k_1}=*$ by \eqref{eq_Akt}.
\end{itemize}
In other words, C$2'$ and C$3$ hold.

\end{proof}

\begin{example} \label{exam3} For a $(6,3,6)$ caching system, Ali-Niesen scheme with $t=3$ can be depicted by the following PDA.
\begin{align}
\mathbf{D}^{6,3}=\left[
  \begin{array}{cccccc}
    * & * & * & 0 & 1 & 2 \\
    * & * & 0 & * & 3 & 4 \\
    * & * & 1 & 3 & * & 5 \\
    * & * & 2 & 4 & 5 & * \\
    * & 0 & * & * & 6 & 7 \\
    * & 1 & * & 6 & * & 8 \\
    * & 2 & * & 7 & 8 & * \\
    * & 3 & 6 & * & * & 9 \\
    * & 4 & 7 & * & 9 & * \\
    * & 5 & 8 & 9 & * & * \\
    0 & * & * & *& 10 & 11 \\
    1 & * & * & 10 & * & 12 \\
    2 & * & * & 11 & 12 & * \\
    3 & * & 10 & * & * & 13 \\
    4 & * & 11 & * & 13 & * \\
    5 & * & 12 & 13 & * & * \\
    6 & 10 & * & * & * & 14 \\
    7 & 11& *& * & 14 & * \\
    8 & 12 & * & 14 & * & * \\
    9 & 13 & 14 & * & * & * \\
  \end{array}
\right]\label{eq_C}
\end{align}
\end{example}

In the remainder of this section, we show the optimality of the PDA for Ali-Niesen scheme. Before that,
we give two useful lemmas to characterize the relationship of parameters of a regular PDA.

\begin{lemma}\label{lem1} For any given $g$-$(K,F,Z,S)$ PDA, then the rate of the corresponding caching scheme is given by
\begin{align}
R=\frac{S}{F}=\frac{K\left(1-\frac{Z}{F}\right)}{g}\label{eq_R}
\end{align}
Moreover, $g$ must satisfy
\begin{align}
g\leq K\frac{Z}{F}+1\label{eq_g}
\end{align}
where the equality holds for $g\ge 2$ only if
each row has exactly $KZ/F$ specific symbols $``*"$.
\end{lemma}

\begin{proof} To prove the rate given by \eqref{eq_R}, let us count the number of the integers in PDA
  in two different manners. On one hand, since each column has $F-Z$ integers, there are totally $K(F-Z)$ integers in
  all the  $K$ columns. On the other hand, since each integer occurs $g$ times,  the total number of all the $S$ integers is $Sg$. Hence,
 \begin{align}
 Sg=K(F-Z)\notag
 \end{align}
 which results in \eqref{eq_R}.

To verify \eqref{eq_g}, denote $t_j$ the number of symbol $``*"$ in row $j$, $j\in[0,F)$.
Since each integer occurs $g$ times in the PDA, each row then contains at least $g-1$ $``*"$s by C$3$-b, which immediately indicates $t_j\ge g-1$ for all $j\in [0,F)$. Note that
totally there are $KZ$ $``*"$s. Thus, we have
\begin{align}
F(g-1)\le \sum_{j=0}^{F-1} t_j= KZ\notag
\end{align}
which gives \eqref{eq_g} with equality holding only if each row has exactly $g-1=KZ/F$ symbols $``*"$.
\end{proof}

\begin{lemma}\label{lem9} Given positive integers $K,F$, and  $g$ s.t., $K\geq g\geq 2$, if an $F\times K$ array $\mathbf{P}$ whose entries consist of  a specific symbol $``*"$ and
some nonnegative integers satisfying \rm{C$2'$, C$3$}, \emph{ and
\begin{enumerate}
  \item [C$4$.] Each row has exactly $g-1$ $``*"$s,
\end{enumerate}}
 \begin{flushleft}
 then $F\geq {K\choose g-1}$. \end{flushleft}
\end{lemma}

\begin{proof}
We prove this lemma by  the induction on the integer $g\ge 2$.

When $g=2$, it follows from C4 that each row has one $``*"$ and $K-1$ integers.
Note the following fact:
\begin{itemize}
\item [] \textbf{Fact 1}. \emph{The conditions
C2$'$, C3, and C4 are still satisfied   after the exchange of two rows/colums in the $F\times K$ array $\mathbf{P}$.}
\end{itemize}
So, we can always assume that the entry in the first row and the first column is $``*"$. Then, the other ones in the first row are integers, which have to be $K-1$ distinct integers further by C$3$-a. Note that each of the $K-1$ integers occurs exactly $g=2$ times. Therefore  by C$3$-b or \eqref{Eqn_Matrix_1}, there are at least one $``*"$ in all the columns $1,\cdots, K-1$. Together with the $``*"$ in the first column, there are at least $K$~$``*"$s in this array. Since there is only one $``*"$ in each row, we conclude that, there are at least $K$ rows, $i.e.$
\begin{align}
F\geq K={K\choose g-1}\notag
\end{align}
holds for $g=2$ and all $K\ge g$.

Suppose that the claim holds for  $g=n$ and all  $K\ge g$, $i.e.$ for $g=n$ and $K\ge n$, we have
  \begin{align}
  F\geq {K\choose n-1}\label{Eqn_Asp}
  \end{align}
if C$2'$, C$3$, and C$4$  hold.

Let $g=n+1$ and  $K\ge n+1$. If an $F\times K$ array $\mathbf{P}$ satisfies C$4$, there are totally $F(g-1)$ $``*"$s in  $\mathbf{P}$. Then, in average, there are $F(g-1)/K$ $``*"$s in each column. Thus, there exists a column having at most  $F(g-1)/K$  $``*"$s.  Based on Fact 1, we
can always transform the original array $\mathbf{P}$ into another  $F\times K$ array $\mathbf{P}'$ such that $\mathbf{P}'$
is of form
\begin{align*}
  \mathbf{P}'=\left[\begin{array}{c|ccc}
                     * & ~~ & ~~ & ~~ \\
                     \vdots &  & \mathbf{P}_1' &  \\
                     * &  &  &  \\\hline
                     a_0 &  &  &  \\
                     \vdots &  & \mathbf{P}_2' &  \\
                     a_{m-1} &  &  &
                   \end{array}
  \right]
  \end{align*}
  where $a_0,\cdots,a_{m-1}$ are integers, $\mathbf{P}_1'$ and $\mathbf{P}_2'$ are $F_1\times (K-1)$ array and $m\times (K-1)$ array respectively, the integer
  $m=F-F_1$, and
  \begin{align}
  F_1\leq \frac{F(g-1)}{K}=\frac{F n}{K}\label{eq_f1}
  \end{align}

 Denote the sets of integers in $\mathbf{P}_1'$ and $\mathbf{P}_2'$ by $\mathcal{P}_1'$ and $\mathcal{P}_2'$ respectively.
 Firstly,  we know from C$2'$ that $a_j$ appears $g-1\ge 1$ times in $\mathcal{P}_1'\cup \mathcal{P}_2'$ for any $j\in [0,m)$. Further,
 we have
 \begin{eqnarray}
 \{a_0,\cdots,a_{m-1}\}\cap \mathcal{P}_2'=\emptyset\label{Eqn_subset-2}
 \end{eqnarray}
and thus
 \begin{eqnarray}
 \{a_0,\cdots,a_{m-1}\}\subseteq \mathcal{P}_1'\label{Eqn_subset-1}
 \end{eqnarray}
 since otherwise if $a_j$ occurs in row $j'$ of $\mathbf{P}_2'$ for some integers $j'\in [0,m)$, then the element $a_{j'}$
 has to be ``$*$" by C$3$-b, a contradiction.
 Secondly, suppose that there is an integer $b$ in the $j$th ($j\in [0, F_1)$) row of $\mathbf{P}_1'$ but $b\notin\{a_0,\cdots,a_{m-1}\}$.  According to\
 C$2'$, $b$ occurs $g=n+1$ times in $\mathbf{P}'$. By C$3$-b or \eqref{Eqn_Matrix_1}, there are at least $n$ $``*"$s in row $j$ of $\mathbf{P}_1'$. Together with the $``*"$ in the first column, there are at least $n+1$ $``*"$s in row $j$ of $\mathbf{P}'$, which contradicts C$4$. Therefore, we have
  \begin{align*}
  \mathcal{P}_1'\subseteq\{a_1,\cdots,a_{m}\}\notag
  \end{align*}
 By means of  \eqref{Eqn_subset-2}, \eqref{Eqn_subset-1}, and the above equation, we arrive at
 \begin{eqnarray*}
  \mathcal{P}_1'=\{a_1,\cdots,a_{m}\}\label{eq_Pa}\ \hbox{and}\
  \mathcal{P}_1'\cap \mathcal{P}_2'=\emptyset\label{eq_Pb}
 \end{eqnarray*}

In addition, it follows from the above discussion that
\begin{itemize}
\item Each integer $a_j$ appears $n$ times in $\mathbf{P}_1'$ since it occurs $n+1$ times in $\mathbf{P}'$;
\item $``*"$ appears $n-1$ times in each row of  $\mathbf{P}_1'$ since it occurs $n$ times in each row of  $\mathbf{P}'$.
\end{itemize}
This is to say, the  $F_1\times (K-1)$ array $\mathbf{P}_1'$ satisfies C$2'$ and C$4$. Moreover, $\mathbf{P}'$ satisfies C$3$, so does
its sub-array  $\mathbf{P}_1'$. Then applying the assumption in \eqref{Eqn_Asp}, we obtain
\begin{align}
  F_1\geq {K-1\choose n-1}\label{eq_f2}
\end{align}

Finally, combining \eqref{eq_f1} and \eqref{eq_f2}, we have
  \begin{eqnarray*}
  F&\ge& \frac{F_1 K}{n}\\
  &\ge&\frac{K}{n}\cdot{K-1\choose n-1}\\\
  &=&\frac{K}{n}\cdot\frac{(K-1)!}{(n-1)!~(K-n)!}\\
  &=& {K\choose n}
  \end{eqnarray*}
That is, the claim is true for $g=n+1$, which finishes the proof.
\end{proof}

Based on Lemmas \ref{lem1} and \ref{lem9}, we are able to prove the following theorem.

\begin{theorem}\label{thm3}
For a $g$-$(K,F,Z,S)$ PDA $\mathbf{P}$, if $g=KZ/F+1\geq2$, then
$
F\geq{K\choose KZ/F}
$.
\end{theorem}

\begin{proof} By Lemma \ref{lem1},  PDA $\mathbf{P}$ satisfies C$4$. Then, the desired result follows
directly from Lemma \ref{lem9}.
\end{proof}

Applying Lemma \ref{lem1} and Theorem \ref{thm3} to a $g$-$(K,F,Z,S)$ PDA in place of $Z/F=M/N$,
we have
\begin{eqnarray}
R&=&\frac{K(1-\frac{M}{N})}{g}\label{eq_R1}\\
g&\leq &\frac{KM}{N}+1\label{eq_g2}
\end{eqnarray}
and
\begin{eqnarray}
F&\geq & {K\choose \frac{KM}{N}}\quad\mbox{if}\quad g=\frac{KM}{N}+1\label{eq_F2}
\end{eqnarray}
Recall that the minimal rate $R$ is the primary concern of a  $(K,M,N)$ caching system. According to \eqref{eq_R1}, it is equivalent to maximizing the coding gain $g$,
which is however upper bounded by \eqref{eq_g2}. As for Ali-Niesen scheme, we see from Theorem \ref{lem7} that
\begin{align}
g_{A-N}= \frac{KM}{N}+1,~~~~
F_{A-N}=  {K\choose \frac{KM}{N}}\notag
\end{align}
In this sense, Ali-Niesen scheme is optimal since it achieves the maximal coding gain with the least $F$.

\section{A New Construction Of PDA}\label{sec_nPDA}

Roughly speaking, \eqref{eq_g2} states that the maximal coding gain of a regular PDA can not exceed $KM/N+1$. In fact, a recent result shows that Ali-Niesen scheme achieves the best rate among all un-coded placement caching schemes \cite{optimal}. Further, \eqref{eq_F2} tells us that to achieve the maximal  coding gain, $F$ has to be at least ${K\choose KM/N}$, which increases approximately as  same as $\tiny{(\frac{N}{M})^{K\frac{M}{N}}(\frac{N}{N-M})^{K(1-\frac{M}{N})}}$.
As a result, naturally it is desirable if we can decrease $F$  dramatically with a cost of a slight decrease of coding gain.
In this section, we construct such PDAs for two particular interesting cases: small cache size $M/N=1/q$ and
large cache size $M/N=(q-1)/q$ for $q\geq2$.

Recall that satisfying \eqref{Eqn_Matrix_1} (equivalently C$3$) is the key feature of a PDA. For a given array $\mathbf{A}=[a_{j,k}]$ of size $F\times K$, composed of $``*"$ and integers $0,1,\cdots,S-1$, define
the \textit{placement set} of user $0\le k<K$ to be
\begin{eqnarray*}
\mathcal{A}_{k}=\{j~|~a_{j,k}=*, 0\le j<F\}
\end{eqnarray*}

The following useful lemma about the placement set   comes from a simple observation on \eqref{Eqn_Matrix_1}.

\begin{lemma}\label{Lem_Subarray}
Given a $(K,F,Z,S)$ PDA and any integer $0\le s<S$, assume that $a_{i_0,k_0}=a_{i_1,k_1}=\ldots =a_{i_{g-1},k_{g-1}}=s$
where $0\le i_0,\cdots, i_{g-1}<F$ and  $0\le k_0,\cdots, k_{g-1}<K$ are respectively two groups of distinct integers.
Then,  C$3$ holds if and only if
\begin{eqnarray}\label{PDA-partition-1}
\mathcal{A}_{k_0}\cap\cdots\cap\mathcal{A}_{k_{g-1}}\cap\{i_0,i_1,\ldots,i_{g-1}\}=\emptyset
\end{eqnarray}
and
\begin{eqnarray}\label{PDA-partition-2}
\bigcap_{0\le h\ne l<g} \mathcal{A}_{k_h}~\bigcap~\{i_0,i_1,\ldots,i_{g-1}\}=i_{l}
\end{eqnarray}
for any $l\in [0,g)$.
\end{lemma}

For instance, in the PDA  $\mathbf{A}^{2,2}$ of Example \ref{exam2}, $a_{1,0}=a_{2,2}=a_{0,5}=0$,     hence the set $\{i_0,i_1,i_2\}=\{1,2,0\}$ and the placement sets $\mathcal{A}_0=\{0,2\}$, $\mathcal{A}_2=\{0,1\}$,   and $\mathcal{A}_5=\{1,2\}$,  clearly
satisfy \eqref{PDA-partition-1} and \eqref{PDA-partition-2}.

In the sequel, we will make use of the so-called partitions in \cite{TWB,LTP} to generate the placement sets  of our desired PDA.
Given an integer $s\in[0,q^n)$ where $n\in\mathbb{N}^+$, with $s=\sum_{l=0}^{n-1}s_lq^l$ for integers $s_l\in[0,q)$, we refer to $s=(s_{n-1},\cdots, s_0)_q$
as the {\em $q$-ary representation\/} of $s$.
There are $n+1$ partitions of $[0,q^n)$, \emph{i.e.}, for $u=0,1,\cdots, n-1$,
\begin{eqnarray}\label{Eqn_Partition_3}
V_{u,v}=\{(s_{n-1},\cdots,s_0)_q\ |\ s_{u}=v\}, ~0\le v<q
\end{eqnarray}
and for $u=n$,
\begin{eqnarray}\label{Eqn_Partition_4}
V_{u,v}=\{(s_{n-1},\cdots,s_0)_q\ |\ \sum_{l=0}^{n-1}s_l=v\},~0\le v<q
\end{eqnarray}
where the sum is performed under modulo $q$.

\subsection{New Construction For $M/N=1/q$}

Set $n=m$ and define $\mathcal{F}_A\overset{\triangle}{=}\{0,1,\cdots,q^m-1\}$.  We construct a $q^m\times q(m+1)$ array $\mathbf{A}^{q,m}$ with $q(m+1)$ users by directly utilizing the partitions in
\eqref{Eqn_Partition_3} and \eqref{Eqn_Partition_4} as the placement sets, \emph{i.e.},
the placement set $\mathcal{A}_{k}$ of user $k=uq+v$ is
\begin{eqnarray}
\mathcal{A}_k=V_{u,v}, 0\le u\le m,\,0\le v<q\notag
\end{eqnarray}

Firstly, C$1$ is obvious with $Z=q^{m-1}$, \emph{i.e.}, the symbol $``*"$ appears $q^{m-1}$ times in each column of $\mathbf{A}^{q,m}$.

Next, let $l_0,\cdots,l_{m}$ be $m+1$ integers in $[0,q)$ with $l_m\not\equiv l_0+\cdots+l_{m-1}(\bmod\, q)$.
It is easily checked from \eqref{Eqn_Partition_3} and \eqref{Eqn_Partition_4} that
\begin{eqnarray}
V_{0,l_0}\cap\cdots\cap V_{m,l_m}&=&\emptyset\label{Eqn_Partition_insert1}
\end{eqnarray}
and
\begin{eqnarray}\label{Eqn_Partition_insert2}
\Lambda_{u}\triangleq\mathop\bigcap\limits_{0\le r\ne u\le m} V_{r,l_r}= \left\{
\begin{array}{ll}
(l_{m-1},\cdots,l_0)_q, & \mathrm{if}~u=m\\
(l_{m-1},\cdots,l_{u+1},l_m-\sum_{0\le l\ne u<m}l_l, l_{u-1},\cdots, l_0)_q,& \mathrm{if}~0\le u<m
\end{array}\right.\end{eqnarray}
where the additions and subtractions are performed under modulo $q$.

Consider the $(m+1)\times (m+1)$ subarray formed by rows $\Lambda_u$ defined in
\eqref{Eqn_Partition_insert2}  and
columns $k=uq+l_{u}$ for all $0\le u\le m$. Then,
this subarray satisfies \eqref{PDA-partition-1}
and \eqref{PDA-partition-2} by \eqref{Eqn_Partition_insert1} and \eqref{Eqn_Partition_insert2}. Clearly, there are $(q-1)q^{m}$ such subarrays.
Therefore, if we assign an unique integer $s\in [0,(q-1)q^{m})$ to each
subarray such that  in this subarray $a_{\Lambda_u,k}=s$, for each $k=uq+l_u$ with $0\le u\le m$, then C$3$ holds for the array $\mathbf{A}^{q,m}$ by Lemma \ref{Lem_Subarray}.

There exist $((q-1)q^{m})!$ distinct assignment methods,
one of which is given in Construction A.

\textbf{Construction A}:  Let $j=(j_{m-1},\cdots,j_{0})_q\in \mathcal{F}_A$ and $k=u q+v\in\mathcal{K}$ where $m\in\mathbb{N}^+$, $0\le u\le m$, and $0\le v<q$.
Define the entry in row $j$ and column $k$ of a $q^m\times q(m+1)$ array $\mathbf{A}^{q,m}$ by
\begin{itemize}
\item $0\le u< m$
\begin{eqnarray}\label{Eqn_New_PDA_1}
a_{j,k}=\left\{
\begin{array}{ll}
*, & \textrm{if}~j_{u}=v\\
(j_{u}-v-1,j_{m-1},\cdots,j_{u+1},v,j_{u-1},\cdots,j_0)_q, & \textrm{if}~j_{u}\neq v
\end{array}
\right.
\end{eqnarray}
\item $u=m$
\begin{eqnarray}\label{Eqn_New_PDA_2}
a_{j,k}=\left\{
\begin{array}{ll}
*, & \textrm{if}~j_0+\cdots+j_{m-1}=v\\
(v-\sum_{l=0}^{m-1} j_l-1,j_{m-1},\cdots,j_0)_q, & \textrm{if}~j_0+\cdots+j_{m-1}\neq v
\end{array}
\right.
\end{eqnarray}
\end{itemize}
where all additions and subtractions are performed under modulo $q$.

\begin{theorem}\label{thm2}
Given $q,m\in\mathbb{N}^+$, $q\geq2$, the array $\mathbf{A}^{q,m}$ given in \eqref{Eqn_New_PDA_1} and \eqref{Eqn_New_PDA_2} is an $(m+1)$-$(q(m+1),q^m,q^{m-1},q^{m+1}-q^m)$ PDA with rate $R=q-1$.
\end{theorem}
\begin{proof} It is sufficient to verify C$1$, C$2'$, and C$3$.

Note from \eqref{Eqn_New_PDA_1} and \eqref{Eqn_New_PDA_2}
that for any integer $s\in [0,q^{m+1}-q^m)$, $i.e.$, $s=(s_{m},s_{m-1},\cdots,s_0)_q$ with $0\le s_m<q-1$,
it appears in row $j\in\mathcal{F}_A$ and column $k=uq+v\in\mathcal{K}$,
if and only if
\begin{itemize}
\item when $0\le u<m$,
\begin{eqnarray}
j&=&(s_{m-1},\cdots,s_{u+1},s_{u}+s_m+1, s_{u-1},\cdots, s_0)_q\label{eq_i1}\\
k&=&uq+s_{u}\label{eq_i11}
\end{eqnarray}
\item when $u=m$,
\begin{eqnarray}
j&=&(s_{m-1},\cdots,s_0)_q\label{eq_i2} \\
k&=&uq+s_0+\cdots+s_m+1\label{eq_i22}
\end{eqnarray}
\end{itemize}

Then, each integer $s\in [0,q^{m+1}-q^m)$ occurs $m+1$ times in the array $\mathbf{A}^{q,m}$
since for any fixed $u$, there is exactly one $j$ and $k$,
such that $a_{j,k}=s$. That is, C$2'$ holds.

Further, let $l_0=s_0,\cdots, l_{m-1}=s_{m-1}$, and
$l_m=s_m+\cdots+s_0+1$, then we see that \eqref{eq_i1}, \eqref{eq_i2} are equivalent to
\eqref{Eqn_Partition_insert2}, and \eqref{eq_i11}, \eqref{eq_i22} are equivalent to
$k=uq+l_u$.  Then, C$1$ and C$3$ hold as well based on
the above discussion. This finishes the proof.

\end{proof}

\begin{example}\label{exam4} The array in \eqref{eq_A} in Example \ref{exam2} is in fact generated by setting $q=2,m=2$. Table \ref{table2} presents the binary representation form of \eqref{Eqn_New_PDA_1}-\eqref{Eqn_New_PDA_2}. Clearly, the corresponding $4\times 6$ array is $\mathbf{A}^{2,2}$ in \eqref{eq_A}.
 \begin{table*}[!htp]
 \centering
\caption{Binary representation of $\mathbf{A}^{2,2}$ with $m=2$ and $q=2$}\label{table2}
 \normalsize{\begin{tabular}{ccccccc}
     \bottomrule
\small{$(j_1,j_0)_2$}\Large{$\backslash$}\small{$(u,v)$} & $(0,0)$&$(0,1)$&$(1,0)$&$(1,1)$&$(2,0)$&$(2,1)$\\\hline
  $(0,0)_2 $ &$ * $&$(0,0,1)_2 $ &$ *$ &$ (0,1,0)_2$ & $*$ &$(0,0,0)_2 $ \\
 $(0,1)_2 $ &$ (0,0,0)_2 $& $*$ &$ *$ &$ (0,1,1)_2$ &$(0,0,1)_2$  &$ * $\\
 $(1,0)_2 $ & $*$ & $(0,1,1)_2 $& $(0,0,0)_2$ &$ * $& $(0,1,0)_2$ &$ *$ \\
 $(1,1)_2 $& $(0,1,0)_2 $& $*$ & $(0,0,1)_2$ & $ *$& $* $&$(0,1,1)_2$\\
\toprule
  \end{tabular}}
\end{table*}
\end{example}

%
%
%

\begin{example}\label{exam5} Let $q=3,m=2$, then Table \ref{table3} gives the ternary representation of the new construction.
\begin{table*}[!htp]
 \centering
\caption{Ternary representation of $\mathbf{A}^{3,2}$ with $m=2$ and $q=3$}\label{table3}
 \normalsize{\begin{tabular}{cccccccccc}
     \bottomrule
 \small{$(j_1,j_0)_3$}\Large{$\backslash$}\small{$(u,v)$}  & $(0,0)$ &$(0, 1)$ &$ (0,2)$ &$ (1,0)$ &$(1, 1)$ &$(1, 2)$&$(2,0)$&$(2,1)$&$(2,2)$ \\\hline
 $(0,0)_3$  &$ *$ &$(1,0,1)_3$  & $(0,0,2)_3$& $*$ &$(1,1,0)_3$&$(0,2,0)_3$ &$*$&$(0,0,0)_3$&$(1,0,0)_3$  \\
 $(0,1)_3$  & $(0,0,0)_3 $& $* $& $(1,0,2)_3$ &$ *$ &$(1,1,1)_3$  & $(0,2,1)_3$ &$(1,0,1)_3$ &$*$&$(0,0,1)_3$\\
 $(0,2)_3$  & $(1,0,0)_3$ &$ (0,0,1)_3 $&$ *$ & $*$ &$ (1,1,2)_3$ &$ (0,2,2)_3$&$(0,0,2)_3$ &$(1,0,2)_3$&$*$ \\
 $(1,0)_3$ & $* $& $(1,1,1)_3 $& $(0,1,2)_3$ &  $(0,0,0)_3$& $*$ &$(1,2,0)_3$& $(1,1,0)_3$&$*$&$(0,1,0)_3$\\
 $(1,1)_3$& $(0,1,0)_3$&  $ *$&  $(1,1,2)_3$&   $(0,0,1)_3$&  $ *$&  $(1,2,1)_3$&  $ (0,1,1)_3$& $(1,1,1)_3 $& $  *$\\
 $(1,2)_3$&$(1,1,0)_3$& $  (0,1,1)_3$&  $ *$& $(0,0, 2)_3$&  $ *$&$ (1,2,2)_3 $&  $ *$&   $(0,1,2)_3$& $ (1,1,2)_3$\\
 $(2,0)_3$&$*$& $ (1,2,1)_3$&$   (0,2,2)_3$&  $ (1,0,0)_3$&  $ (0,1,0)_3$& $  *$&  $(0,2,0)_3$& $(1,2,0)_3$& $ *$\\
 $(2,1)_3$& $(0,2,0)_3$&  $ *$&  $(1,2,2)_3$&  $(1,0,1)_3$& $  (0,1,1)_3$&  $ *$&  $ *$&  $(0,2,1)_3$& $(1,2,1)_3$\\
 $(2,2)_3$&$(1,2,0)_3$&  $ (0,2,1)_3$& $  *$&  $(1,0,2)_3$&   $(0,1,2)_3$& $  *$&  $(1,2,2)_3$& $  *$&  $ (0,2,2)_3$\\
\toprule
  \end{tabular}}
\end{table*}
Then, the corresponding $9\times 9$ array is
\begin{align*}
\mathbf{A}^{3,2}=\left[
  \begin{array}{ccccccccc}
 *&  10&   2&   *&  12&   6&   *&   0&   9\\
 0&   *&  11&   *&  13&   7&  10&   *&   1\\
 9&   1&   *&   *&  14&   8&   2&  11&   *\\
 *&  13&   5&   0&   *&  15&  12&   *&   3\\
 3&   *&  14&   1&   *&  16&   4&  13&   *\\
12&   4&   *&   2&   *&  17&   *&   5&  14\\
 *&  16&   8&   9&   3&   *&   6&  15&   *\\
 6&   *&  17&  10&   4&   *&   *&   7&  16\\
15&   7&   *&  11&   5&   *&  17&   *&   8
  \end{array}
\right]
\end{align*}
\end{example}
Comparing $\mathbf{A}^{2,2}$ in \eqref{eq_A} and $\mathbf{D}^{6,3}$  in \eqref{eq_C}, they both support $6$ users with $M/N=1/2$.
While having $1$ less coding gain ($3$ versus $4$) and thus a larger rate ($1$ versus $3/4$) compared to $\mathbf{D}^{6,3}$,
$\mathbf{A}^{2,2}$ has the advantage that the number of packets needed to be split into is much smaller
($4$ versus $20$).

\begin{remark}
During this paper was under review, a coded caching scheme  was proposed in \cite{TR} based on the resolvable block design. In fact, it is equivalent to
our PDA specified in \eqref{Eqn_New_PDA_1} and \eqref{Eqn_New_PDA_2}.

\end{remark}

\subsection{New Construction For $M/N=(q-1)/q$}\label{sec_B}

Let $n=m+1$. Then,
\begin{eqnarray}\label{Eqn_Partition_31}
V_{u,v}=\{(s_m,\cdots,s_0)\ |\ s_m\in[0,q-1),\, s_{u}=v\}, ~0\le v<q
\end{eqnarray}
where $u\in[0,m)$ and
\begin{eqnarray}\label{Eqn_Partition_41}
V_{m+1,v}=\{(s_{m},\cdots,s_0)\ |\ s_m\in[0,q-1),\,\sum_{l=0}^{m}s_l =v\},~0\le v<q
\end{eqnarray}
are $m+1$ partitions of $\mathcal{F}_{B}\overset{\triangle}{=}\{0,\cdots, (q-1)q^{m}-1\}$, where the sum in \eqref{Eqn_Partition_41} is performed under modulo $q$.

For $k=uq+v\in\mathcal{K}$ with $0\le u\le m$ and $0\le v<q$, let the placement set $\mathcal{B}_k$ of user $k$ be
\begin{eqnarray}
\mathcal{B}_k=\left\{
\begin{array}{ll}
\mathcal{F}_B\backslash V_{u,v}, & \mathrm{if}~u\in[0,m),v\in[0,q)\\
\mathcal{F}_B\backslash V_{m+1,v-1}, & \mathrm{if}~u=m, v\in[0,q)
\end{array}
\right.\label{Eqn:B_definition}
\end{eqnarray}
Straightforwardly, the symbol $``*"$ appears $(q-1)^2q^{m-1}$ times in each column of $\mathbf{B}^{q,m}$. That is,
C$1$ holds with $Z=(q-1)^2q^{m-1}$.

Next,
let $l_0,\cdots,l_{m-1}$ be arbitrary $m$ integers in $[0,q)$.
By convenience, denote $l_{m}\equiv l_0+\cdots+l_{m-1}(\bmod\,q)$.
Define
\begin{eqnarray*}
\Omega_u &=&\{uq+v~|~0\le v<q\}\setminus \{uq+l_u\}\\
\Omega &=& \bigcup\limits_{l=0}^m \Omega_l
\end{eqnarray*}

Then, it can be verified from \eqref{Eqn_Partition_31} and \eqref{Eqn_Partition_41} that
\begin{align}
\bigcap_{r\in\Omega}\mathcal{B}_{r}=\emptyset\label{eqn:B_empty}
\end{align}
and
\begin{align}
\Lambda_{k}\overset{\triangle}{=}\bigcap_{r\in\Omega\backslash\{k\}}\mathcal{B}_{r}=\left\{\begin{array}{ll}
                                                                  (l_u-v-1,l_{m-1},\cdots,l_{u+1},v,l_{u-1},\cdots,l_0)_q,  &\mbox{if}~u\in[0,m) \\
                                                                 (v-l_m-1,l_{m-1},\cdots,l_0)_q,  &\mbox{if}~u=m
                                                                 \end{array}
\right.\label{eqn:B_single}
\end{align}
where $k=uq+v\in \Omega$ with $0\le u\le m$, $0\le v<q$, and the subtractions are performed under modulo $q$,  since
\begin{align}
\bigcap_{r\in\Omega_u}\mathcal{B}_{r}=\left\{\begin{array}{ll}
                                                 V_{u,l_u} &\mbox{if}~u\in[0,m)  \\
                                                 V_{m+1,l_m-1} &\mbox{if}~u=m
                                               \end{array}
\right.\notag
\end{align}
and
\begin{align}
\bigcap_{r\in\Omega_u\backslash\{uq+v\}}\mathcal{B}_{r}=\left\{\begin{array}{ll}
                                                     V_{u,l_u}\cup V_{u,v}  &\mbox{if}~u\in[0,m), v\neq l_u  \\
                                                     V_{m+1,l_m-1}\cup V_{m+1,v-1}  &\mbox{if}~u=m, v\neq l_u
                                                    \end{array}
\right.\notag
\end{align}


Consider the $(q-1)(m+1)\times (q-1)(m+1)$ subarray formed by rows $\Lambda_k$ and columns $k$ for $k\in\Omega$. The subarray satisfies  \eqref{PDA-partition-1} and \eqref{PDA-partition-2} by \eqref{eqn:B_empty} and \eqref{eqn:B_single}.
Thus, if we assign an unique integer $s\in [0,q^{m})$ to each
subarray such that  in this subarray $a_{\Lambda_{k},k}=s$ for all $k\in\Omega$, then C$3$ holds.
There are $(q^{m})!$ distinct assignment methods,
one of which is given in Construction B.

\textbf{Construction B}: Let $ j=(j_{m},\cdots,j_{0})_q\in\mathcal{F}_B$ and $ k=uq+v\in\mathcal{K}$ where $m\in\mathbb{N}^+$,$0\leq u\leq m$, and $0\leq v<q$.
Define the entry in row $j$ and column $k$ in a $(q-1)q^{m}\times q(m+1)$ array $\mathbf{B}^{q,m}$ by
\begin{itemize}
\item $0\le u< m$
\begin{eqnarray}\label{Eqn_New_PDA_B1}
b_{j,k}=\left\{
\begin{array}{ll}
(j_{m-1},\cdots,j_{u+1},j_{u}+j_{m}+1,j_{u-1},\cdots,j_0)_q, & \textrm{if}~j_{u}=v\\
*,& \textrm{if}~j_{u}\neq v
\end{array}
\right.
\end{eqnarray}
\item $u=m$
\begin{eqnarray}\label{Eqn_New_PDA_B2}
b_{j,k}=\left\{
\begin{array}{ll}
(j_{m-1},\cdots,j_0)_q,& \textrm{if}~j_0+\cdots+j_{m-1}+j_m=v-1\\
*, & \textrm{if}~j_0+\cdots+j_{m-1}+j_m\neq v-1
\end{array}
\right.
\end{eqnarray}
where all additions and subtractions are performed under modulo $q$.
\end{itemize}

\begin{theorem}\label{thm4}
Given $q,m\in\mathbb{N}^+$, $q\geq2$, the array $\mathbf{B}^{q,m}$ given in \eqref{Eqn_New_PDA_B1} and \eqref{Eqn_New_PDA_B2} is a $(q-1)(m+1)$-$(q(m+1),(q-1)q^{m},(q-1)^2q^{m-1},q^{m})$ PDA with rate $R=1/(q-1)$.
\end{theorem}
\begin{proof} It suffices to verify C$1$, C$2'$, and C$3$.

Given an integer $s\in [0,q^{m})$, $i.e.$, $s=(s_{m-1},\cdots,s_0)_q$,
assume $b_{j,k}=s$ where $j=(j_m,\cdots,j_0)_q\in\mathcal{F}_B$, $k=uq+v\in\mathcal{K}$. According to \eqref{Eqn_New_PDA_B1} and \eqref{Eqn_New_PDA_B2},
\begin{itemize}
\item when $0\le u<m$ and $0\leq j_m<q-1$,
\begin{eqnarray}
j&=&(j_m,s_{m-1},\cdots,s_{u+1},s_u-j_m-1, s_{u-1},\cdots, s_0)_q\label{eq_g1}\\
k&=&uq+s_u-j_m-1\label{eq_g11}
\end{eqnarray}
\item when $u=m$ and $0\le j_m< q-1$,
\begin{eqnarray}
j&=&(j_m,s_{m-1},\cdots,s_0)_q\label{eqn_g2} \\
k&=&uq+s_0+\cdots+s_{m-1}+j_m+1\label{eq_g22}
\end{eqnarray}
\end{itemize}

Then, each integer $s\in [0,q^{m})$ occurs $(q-1)(m+1)$ times in the array $\mathbf{B}^{q,m}$
since for each $u\in [0,m]$ and $j_m \in [0, q-1)$, there exists
exactly one pair of integers $j$ and $k$ such that $b_{j,k}=s$. That is, C$2'$ holds.

Furthermore, let $l_0=s_0,l_1=s_1,\cdots,l_{m-1}=s_{m-1}$, then \eqref{eq_g1}-\eqref{eq_g22}  are equivalent to
\eqref{eqn:B_single} and $k=uq+v\in \Omega$.
 Then, C$1$ and C$3$ hold based on
the above discussion. This completes the proof.
\end{proof}

\begin{example}\label{exam6} Let $q=3,m=2$, then Table \ref{table4} gives the ternary representation of the new construction.
\begin{table*}[!htp]
 \centering
\caption{Ternary representation of $\mathbf{B}^{3,2}$ with $m=2$ and $q=3$}\label{table4}
 \normalsize{\begin{tabular}{cccccccccc}
     \bottomrule
 \small{$(j_2,j_1,j_0)_3$}\Large{$\backslash$}\small{$(u,v)$}
 & $(0,0)$ &$(0, 1)$ &$ (0,2)$ &$ (1,0)$ &$(1, 1)$ &$(1, 2)$&$(2,0)$&$(2,1)$&$(2,2)$ \\\hline
 $(0,0,0)_3$  &$ (0,1)_3$ &$*$  & $*$& $(1,0)_3$ & $*$&$*$ &$*$&$(0,0)_3$&$*$  \\
 $(0,0,1)_3$  & $*$& $(0,2)_3 $& * &$ (1,1)_3$ &$*$  & $*$ &$*$ &$*$&$(0,1)_3$\\
 $(0,0,2)_3$  & $*$ &$*$& $(0,0)_3$ & $(1,2)_3$ &$*$ & $*$&$(0,2)_3$ &$*$&$*$ \\
 $(0,1,0)_3$ & $(1,1)_3 $& $* $&  $* $&   $* $& $(2,0)_3$&$*$ & $*$ &$*$&$(1,0)_3$\\
 $(0,1,1)_3$& $*$ &$(1,2)_3$&  $ *$&     $*$ &   $(2,1)_3$& $ *$&    $(1,1)_3 $& $ *$&$  *$\\
 $(0,1,2)_3$&$ *$& $ *$&$  (1,0)_3$&  $ *$&   $ (2,2)_3 $&$ *$&  $ *$&   $(1,2)_3$&$ *$\\
 $(0,2,0)_3$& $ (2,1)_3$&$*$&$*$&  $*$&$*$& $ (0,0)_3$&  $(2,0)_3$& $  *$&  $ *$\\
 $(0,2,1)_3$& $ *$&$(2,2)_3$&    $ *$& $ *$& $ *$&$  (0,1)_3$&    $ *$&  $(2,1)_3$& $ *$\\
 $(0,2,2)_3$&$ *$& $  *$& $(2,0)_3$&   $ *$& $  *$&    $(0,2)_3$&  $ *$ & $  *$& $(2,2)_3$ \\
 $(1,0,0)_3$  &$ (0,2)_3$ &$*$  & $*$& $(2,0)_3$ & $*$&$*$ &$*$&$*$&$(0,0)_3$  \\
 $(1,0,1)_3$  & $*$& $(0,0)_3 $& * &$ (2,1)_3$ &$*$  & $*$  &$(0,1)_3$&$*$&$*$\\
 $(1,0,2)_3$  & $*$ &$*$& $(0,1)_3$ & $(2,2)_3$ &$*$ & $*$&$*$&$(0,2)_3$ &$*$ \\
 $(1,1,0)_3$ & $(1,2)_3 $& $* $&  $* $&   $* $& $(0,0)_3$&$*$ & $(1,0)_3$&$*$&$*$ \\
 $(1,1,1)_3$& $*$ &$(1,0)_3$&  $ *$&     $*$ &   $(0,1)_3$& $ *$&$  *$&    $(1,1)_3 $& $ *$\\
 $(1,1,2)_3$&$ *$& $ *$&$  (1,1)_3$&  $ *$&   $ (0,2)_3 $&$ *$&  $ *$& $ *$ & $(1,2)_3$\\
 $(1,2,0)_3$& $ (2,2)_3$&$*$&$*$&  $*$&$*$& $ (1,0)_3$&  $ *$&  $(2,0)_3$& $  *$\\
 $(1,2,1)_3$& $ *$&$(2,0)_3$&    $ *$& $ *$& $ *$&$  (1,1)_3$& $ *$&    $ *$&  $(2,1)_3$\\
 $(1,2,2)_3$&$ *$& $  *$& $(2,1)_3$&   $ *$& $  *$&    $(1,2)_3$& $(2,2)_3$&  $ *$ & $  *$ \\
\toprule
  \end{tabular}}
\end{table*}
Then, the corresponding $18\times 9$ array is
\begin{align*}
\mathbf{B}^{3,2}=\left[
  \begin{array}{ccccccccc}
1&   *&   *&   3&  *&    *&   *&    0&   *\\
*&   2&   *&   4&  *&    *&   *&    *&   1\\
*&   *&   0&   5&  *&    *&   2&    *&   *\\
4&   *&   *&   *&  6&    *&   *&    *&   3\\
*&   5&   *&   *&  7&    *&   4&    *&   *\\
*&   *&   3&   *&  8&    *&   *&    5&   *\\
7&   *&   *&   *&  *&    0&   6&    *&   *\\
*&   8&   *&   *&  *&    1&   *&    7&   *\\
*&   *&   6&   *&  *&    2&   *&    *&   8\\
2&   *&   *&   6&  *&    *&   *&    *&    0\\
*&   0&   *&   7&  *&    *&   1&    *&    *\\
*&   *&   1&   8&  *&    *&   *&    2&    *\\
5&   *&   *&   *&  0&    *&   3&    *&    *\\
*& 3  &   *&   *&  1&    *&   *&    4&    *\\
*&   *&   4&   *&  2&    *&   *&    *&    5\\
8&   *&   *&   *&  *&    3&   *&    6&    *\\
*&   6&   *&   *&  *&    4&   *&    *&    7\\
*&   *&   7&   *&  *&    5&   8&    *&    *\\
  \end{array}
\right]
\end{align*}
\end{example}

\begin{remark}
Theorem \ref{thm2} (or Theorem \ref{thm4} respectively) shows that we are able to construct an $F\times K=q^m\times q(m+1)$ (or $(q-1)q^m\times q(m+1)$) PDA, which supports $K=q(m+1)$ users when $M/N=1/q$ (or $(q-1)/q$). For a general  large $K~(K\geq 2q)$, one can construct a PDA to support $K$ users by setting $m=\lceil\frac{K}{q}\rceil-1$, and construct a $q^{\lceil\frac{K}{q}\rceil-1}\times q\lceil\frac{K}{q}\rceil$ (or $(q-1)q^{\lceil\frac{K}{q}\rceil-1}\times q\lceil\frac{K}{q}\rceil$) PDA firstly,  and then delete any $q\lceil\frac{K}{q}\rceil-K$  columns, the resultant rate is not larger than $q-1$ (or $1/(q-1)$).
\end{remark}

%


\section{Performance Analysis}\label{sec_comp}

In this section, we compare the performance of the proposed scheme with the existing ones:  Firstly, we compare the proposed scheme with Ali-Niesen scheme directly; Secondly, for large number of users $K$,  to achieve a coding gain $g$, we investigate the performances of the approaches that grouping users into smaller groups based on Ali-Niesen scheme and the new scheme.


For analysis simplicity, the following lemma is useful.

\begin{lemma}\label{lem_n}
For fixed rational number $M/N\in(0,1)$, let $K\in\mathbb{N}^+$ such that $KM/N\in\mathbb{N}^+$, when $K\rightarrow\infty$,
\begin{align}
{K\choose \frac{KM}{N}} \sim \frac{N}{\sqrt{2\pi M(N-M)}}\cdot e^{K \left(\frac{M}{N}\ln \frac{N}{M}+\left(1-\frac{M}{N}\right)\ln \frac{N}{N-M}\right)}\notag
\end{align}
\end{lemma}
\begin{proof} The well-known Stirling's formula tells us
\begin{align*}
n!\sim\sqrt{2\pi n}\left(\frac{n}{e}\right)^n\notag
\end{align*}
as $n\rightarrow\infty$ for $n\in\mathbb{N}^+$. Therefore, when $n\rightarrow\infty$, we have
\begin{align}
{K\choose \frac{KM}{N}}&=\frac{K!}{\frac{KM}{N}!\left(K\left(1-\frac{M}{N}\right)\right)!}\notag\\
&\sim \frac{\sqrt{2\pi K}\left(\frac{K}{e}\right)^{K}}{\sqrt{2\pi K\frac{M}{N}}\left(\frac{K\frac{M}{N}}{e}\right)^{K\frac{M}{N}}\cdot\sqrt{2\pi K\left(1-\frac{M}{N}\right)} \left(\frac{K\left(1-\frac{M}{N}\right)}{e}\right)^{K\left(1-\frac{M}{N}\right)}}\notag\notag\\
&= \frac{1}{\sqrt{2\pi\frac{M}{N}\left(1-\frac{M}{N}\right)}}\cdot \frac{1}{\left(\frac{M}{N}\right)^{K\frac{M}{N}}\left(1-\frac{M}{N}\right)^{K\left(1-\frac{M}{N}\right)}}\notag\\
&=\frac{N}{\sqrt{2\pi M(N-M)}}\cdot e^{K \left(\frac{M}{N}\ln \frac{N}{M}+\left(1-\frac{M}{N}\right)\ln \frac{N}{N-M}\right)}\notag
\end{align}
\end{proof}


\subsection{Comparison With Ali-Niesen Scheme }

It is easy to observe that, when $K=q(m+1)$, $q,m\in\mathbb{N}^+$, and  $M/N=1/q$ or $(q-1)/q$, the coding gain achieved by Ali-Niesen scheme and ours are
\begin{align}
g_{A-N}=\frac{KM}{N}+1\ \ \ \ \ \hbox{and}\ \ \ \ \ g_{New}=\frac{KM}{N}\notag
\end{align}
 respectively. That is, there is  only  $1$  loss in coding gain, $i.e.$,
 \begin{align}\label{eq_gcomp}
 g_{A-N}-g_{New}=1
 \end{align}
Consequently, the rates of Ali-Niesen scheme and ours are respectively
\begin{align*}
R_{A-N}=\frac{K(1-\frac{M}{N})}{g_{A-N}}\ \ \ \ \ \hbox{and}\ \ \ \ \
R_{New}=\frac{K(1-\frac{M}{N})}{g_{New}}
\end{align*}
Define
\begin{align}
\lambda_{K,\frac{M}{N}}\overset{\triangle}{=}\frac{R_{A-N}}{R_{New}}=\frac{KM/N}{KM/N+1}\label{eqn_lambda}
\end{align}

While on the other hand, Ali-Niesen scheme and ours have to split each file into $F_{A-N}$ packets and $F_{New}$ packets respectively, where
\begin{align}
F_{A-N}&={K\choose \frac{KM}{N}}\notag\\
&\sim \frac{N}{\sqrt{2\pi M(N-M)}}\cdot e^{K \left(\frac{M}{N}\ln \frac{N}{M}+\left(1-\frac{M}{N}\right)\ln \frac{N}{N-M}\right)}\label{eq_Fan}
\end{align}
by Lemma \ref{lem_n}, and
\begin{align}
F_{New}&=\left\{\begin{array}{ll}
                 q^{m} &~~~~\mbox{if~~}\frac{M}{N}=\frac{1}{q}  \\
                 (q-1)q^{m} &~~~~\mbox{if~~}\frac{M}{N}=\frac{q-1}{q}
               \end{array}
\right.\notag\\
&=\left\{\begin{array}{ll}
           \frac{M}{N}\cdot e^{K\frac{M}{N}\ln \frac{N}{M}} & ~~~~\mbox{if~~}\frac{M}{N}=\frac{1}{q} \\
            \frac{M}{N}\cdot e^{K\left(1-\frac{M}{N}\right)\ln\frac{N}{N-M}}& ~~~~\mbox{if~~}\frac{M}{N}=\frac{q-1}{q}
         \end{array}
\right.\label{eq_Fnew}
\end{align}
Define
\begin{align}
\eta_{K,\frac{M}{N}}&\overset{\triangle}{=}\frac{F_{A-N}}{F_{New}}\notag\\
&\sim\left\{\begin{array}{ll}
        \frac{1}{\sqrt{2\pi}}\cdot\left(\frac{N}{M}\right)^{\frac{3}{2}}\cdot\sqrt{\frac{N}{N-M}}\cdot e^{K\left(1-\frac{M}{N}\right)\ln \frac{N}{N-M} }  &~~~~\mbox{if~~}\frac{M}{N}=\frac{1}{q}   \\
       \frac{1}{\sqrt{2\pi}}\cdot\left(\frac{N}{M}\right)^{\frac{3}{2}}\cdot\sqrt{\frac{N}{N-M}}\cdot e^{K\frac{M}{N}\ln \frac{N}{M} }   & ~~~~\mbox{if~~}\frac{M}{N}=\frac{q-1}{q}
         \end{array}
\right.\label{eqn_eta}
\end{align}

Clearly, from \eqref{eqn_lambda} and \eqref{eqn_eta},
\begin{align}
\lim_{K\rightarrow\infty} \lambda_{K,\frac{M}{N}}&=1\label{eq_a2}\\
\lim_{K\rightarrow\infty} \eta_{K,\frac{M}{N}}&=\infty\label{eq_a3}
\end{align}

According to \eqref{eq_gcomp} and \eqref{eq_a2}, compared to Ali-Niesen scheme,  our new scheme achieves a coding gain $1$ less than Ali-Niesen scheme, or in terms of rate, the ratio $\lambda_{K,\frac{M}{N}}$   is approximately $1$ when $K$ becomes large. While  from  \eqref{eq_Fan},   \eqref{eq_Fnew} and \eqref{eqn_eta}, it is clear that, $F_{A-N}$ is of order $O\left(e^{K\cdot\left(\frac{M}{N}\ln \frac{N}{M}+(1-\frac{M}{N})\ln\frac{N}{N-M}\right)}\right)$ and $F_{New}$ is of order $O\left(e^{K\cdot\frac{M}{N}\ln \frac{N}{M}}\right)$ or $O\left(e^{K\left(1-
\frac{M}{N}\right)\ln \frac{N}{N-M}}\right)$ when $M/N$ is $1/q$ or $(q-1)/q$ respectively.   The new construction saves a factor $\eta_{K,\frac{M}{N}}$ of order $O\left(e^{K\left(1-
\frac{M}{N}\right)\ln \frac{N}{N-M}}\right)$ or  $O\left(e^{K\cdot\frac{M}{N}\ln \frac{N}{M}}\right)$, which goes to infinity exponentially with $K$.

We summary the comparisons in Table \ref{table5} and conduct some numerical results for $M/N=1/3$, $1/2$, $2/3$  and some $K$
in Table \ref{table6}.

  \begin{table*}[!htp]\centering\caption{Performance comparison of Ali-Niesen scheme and  new scheme}\label{table5}
        \normalsize{\begin{tabular}{c| c| c}    \bottomrule\bottomrule
        ~~~~~~~~~Performance~~~~~~~~~&Ali-Niesen scheme&   New scheme\\\hline
     \begin{minipage}{0.8cm} \vspace{0.35cm}\centering $g$ \vspace{0.35cm}\end{minipage}&$\frac{KM}{N}+1$&$\frac{KM}{N}$\\\cline{1-3}
                 \begin{minipage}{0.8cm} \vspace{0.35cm}\centering $R$ \vspace{0.35cm}\end{minipage} & $\frac{K\left(1-\frac{M}{N}\right)}{1+KM/N}$ & $\frac{N}{M}-1$\\\cline{1-3}
                \begin{minipage}{3cm} \vspace{0.35cm}\centering $F~(\frac{M}{N}=\frac{1}{q})$ \vspace{0.35cm}\end{minipage} & $\sim\frac{N}{\sqrt{2\pi M(N-M)}}\cdot e^{K \left(\frac{M}{N}\ln \frac{N}{M}+\left(1-\frac{M}{N}\right)\ln \frac{N}{N-M}\right)}$ & $ \frac{M}{N}\cdot e^{K\frac{M}{N}\ln \frac{N}{M}}$\\\hline
             \begin{minipage}{3cm} \vspace{0.35cm}\centering $F~(\frac{M}{N}=\frac{q-1}{q})$ \vspace{0.35cm}\end{minipage}&$\sim\frac{N}{\sqrt{2\pi M(N-M)}}\cdot e^{K \left(\frac{M}{N}\ln \frac{N}{M}+\left(1-\frac{M}{N}\right)\ln \frac{N}{N-M}\right)}$&$\frac{M}{N}\cdot e^{K\left(1-\frac{M}{N}\right)\ln\frac{N}{N-M}}$\\
    \toprule\toprule
        \end{tabular}}

\end{table*}

%

  \begin{table*}[!htp]\centering\caption{Numerical Comparisons of Ali-Niesen scheme and new scheme}\label{table6}
        \normalsize{\begin{tabular}{c|c|cccccc}   \bottomrule\bottomrule
         \multicolumn{2}{c|}{$K$}&$6$& $12$&$18$&$24$&$30$&$36$\\\hline
        \multirow{6}{*}{$\frac{M}{N}=\frac{1}{3}$} &$g_{A-N}$    & 3&  5 &    7  &   9 &   11 &   13 \\
        &$g_{New}$ &   2  &   4 &    6 &    8&    10&    12
  \\\cline{2-8}
        &$R_{A-N}$ &
    1.3333   & 1.6000 &   1.7143 &   1.7778 &   1.8182   & 1.8462 \\
        &$R_{New}$ &2 &2 &2&2&2&2  \\\cline{2-8}
        &$F_{A-N}$ &15 &495 &18564&735471&30045015&    1251677700 \\
        &$F_{New}$ &

           3        &  27 &        243 &       2187 &      19683  &    177147  \\\hline
              \multirow{6}{*}{$\frac{M}{N}=\frac{1}{2}$} & $g_{A-N}$&     4 &    7 &   10 &   13 &   16 &   19 \\
        &$g_{New}$ &  3 &    6 &    9 &   12 &   15 &   18  \\\cline{2-8}
        &$R_{A-N}$ &    0.7500&    0.8571 &   0.9000 &   0.9231&    0.9375&    0.9474 \\
        &$R_{New}$ & 1& 1&1&1&1&1  \\\cline{2-8}
        &$F_{A-N}$ &20 &924 &48620&2704156&   155117520& 9075135300  \\
        &$F_{New}$ &           4   &       32 &        256 &       2048 &      16384 &     131072 \\\hline
              \multirow{6}{*}{$\frac{M}{N}=\frac{2}{3}$} &$g_{A-N}$ &     5  &   9 &   13 &   17 &   21&    25\\
        &$g_{New}$ &   4   &  8  &  12 &   16 &   20 &   24 \\\cline{2-8}
        &$R_{A-N}$ &    0.4000  &  0.4444 &   0.4615  &  0.4706 &   0.4762&    0.4800 \\
        &$R_{New}$ & 0.5& 0.5&0.5&0.5&0.5&0.5  \\\cline{2-8}
        &$F_{A-N}$ & 15 &495 &18564&735471&30045015&    1251677700  \\
        &$F_{New}$ &      6     &     54 &        486 &       4374   &    39366 &     354294  \\\hline
    \toprule\toprule
        \end{tabular}}

\end{table*}
\subsection{Grouping  Based On The New Schemes}

The problem of reducing the number of packets in coded caching was investigated in \cite{finite2016} in another way.
For a $(K,M,N)$ system, where $K$ is assumed to be large enough and have eligible divisibility,
in order to attain a coding gain $g$ with $g\leq KM/N+1$, the key idea in \cite{finite2016} is  to split users into groups of size $K'=(g-1)\lceil N/M\rceil$ and
then implement an Ali-Niesen scheme in a $(K',M',N)$ system for each group, where $M'=N\cdot\frac{1}{\lceil N/M\rceil}\leq M$.
Accordingly, the achievable rate and number of packets are
respectively
\begin{align}
R_{A-N,G}&=\frac{K}{g}\left(1-\frac{1}{\lceil N/M\rceil}\right)\label{eqnR_AN}\\
F_{A-N,G}&={K'\choose g-1}={(g-1)\lceil N/M\rceil\choose g-1}\notag\\
&\sim\frac{\lceil N/M\rceil}{\sqrt{2\pi(\lceil N/M\rceil-1)}}\cdot e^{(g-1)\cdot\left(\ln \lceil N/M\rceil+\left(\lceil N/M\rceil-1\right)\ln \frac{\lceil N/M\rceil}{\lceil N/M\rceil-1}\right)}\label{eqnF_AN}
\end{align}

 The grouping idea is simple but an efficient approach to reduce the number of packets $F$. Motivated by this idea, we propose the following schemes to achieve a coding gain $g$:
\begin{itemize}
  \item When $M/N\leq 1/2$,  let $q=\lceil N/M\rceil$ and $m=g-1$, then implement a caching scheme based on the $g$-$(q(m+1),q^m,q^{m-1},q^{m+1}-q^m)$ PDA as depicted in Theorem \ref{thm2} for each group of size $K_{A}=q(m+1)=g\lceil N/M\rceil$. The resultant achievable rate $R_{A,G}$ and the number of packets $F_{A,G}$ are respectively given by
      \begin{align}
      R_{A,G}&=\frac{K}{K_A}\cdot(q-1)\notag\\
      &=\frac{K}{g}\cdot\left(1-\frac{1}{\lceil N/M\rceil}\right)\label{eqn:R_AG}\\
      F_{A,G}&=q^{m}=\left\lceil N/M\right\rceil^{g-1}=e^{(g-1)\ln \lceil N/M\rceil }\label{eqn:F_AG}
      \end{align}
  Comparing \eqref{eqnR_AN}, \eqref{eqnF_AN} with \eqref{eqn:R_AG}, \eqref{eqn:F_AG}, we get
  \begin{align}
  R_{A,G}&=R_{A-N,G}\label{R:AG}\\
  \frac{F_{A-N,G}}{F_{A,G}}&\sim \frac{\lceil N/M\rceil}{\sqrt{2\pi(\lceil N/M\rceil-1)}}\cdot e^{(g-1)\cdot(\lceil N/M\rceil-1)\ln \frac{\lceil N/M\rceil}{\lceil N/M\rceil-1}}\label{F:AG}
  \end{align}
  \item When $M/N>1/2$, let $q=\lfloor N/(N-M)\rfloor\geq2$,  $m=\lceil g/(q-1)\rceil-1$, then implement a $(q-1)(m+1)$-$\left(q(m+1),(q-1)q^m,(q-1)^2q^{m-1},q^m\right)$ PDA based caching scheme as depicted in Theorem \ref{thm4} for  each group of size $K_B=q(m+1)=q\lceil g/(q-1)\rceil$. The resultant achievable rate $R_{B,G}$ and the number of packets $F_{B,G}$ are respectively given by
      \begin{align}
      R_{B,G}&=\frac{K}{K_B}\cdot \frac{1}{q-1}\notag\\
      &=\frac{K}{(q-1)\left\lceil\frac{g}{q-1}\right\rceil}\cdot\frac{1}{q}\label{eqn:R_BG}\\
      F_{B,G}&=(q-1)q^{\lceil\frac{g}{q-1}\rceil-1}\notag\\
      &\leq (q-1)\cdot \left(q^{\frac{1}{q-1}}\right)^{g}\label{eqn:F_BG}
      \end{align}
Comparing \eqref{eqnR_AN} with \eqref{eqn:R_BG} we obtain
  \begin{align}
  R_{B,G}&\leq R_{A-N,G}\label{R:BG}
  \end{align}
Further, $F_{A-N,G}\sim4^g/\sqrt{8\pi}$  by \eqref{eqnF_AN} and hence with \eqref{eqn:F_BG} we have
  \begin{align}
 \frac{F_{A-N,G}}{F_{B,G}}&\geq \frac{1}{\sqrt{8\pi}}\cdot\frac{1}{q-1}\cdot\left(\frac{4}{q^{\frac{1}{q-1}}}\right)^g\notag\\
 &\geq  \frac{1}{\sqrt{8\pi}}\cdot\frac{1}{q-1}\cdot 2^g\label{F:BG}
  \end{align}
\end{itemize}

From \eqref{R:AG} and \eqref{R:BG}, it is easy to see that the grouping algorithm based on the new scheme achieves at least as well as the scheme that base on Ali-Niesen scheme, while \eqref{F:AG} and \eqref{F:BG} indicate that the grouping algorithm base on the new scheme can save a factor  exponentially increasing with the coding gain $g$. That is, the grouping algorithm based on the new proposed scheme is more efficient on reducing $F$  in contrast to  that based on Ali-Niesen scheme.

\section{Conclusions}\label{sec_conclusion}
In this paper, we defined a new array PDA, which can be used to describe the placement  and delivery schemes in caching system, for example Ali-Niesen caching scheme. Based on a PDA of size $F\times K$, the caching scheme can support $K$ users by dividing each file into $F$ packets. Therefore, the problem of designing a centralized caching scheme can be translated into a problem of designing an appropriate PDA.  Particularly, we established an upper bound on coding gain for all possible regular PDAs and proved that Ali-Niesen scheme achieves the upper bound with the least possible $F$ in all schemes corresponding to regular PDAs. Furthermore, we presented a new construction of PDA for the cases $M/N$ is $1/q$ and $(q-1)/q$ for each integer $q\geq 2$. The new construction  leads to  less order of $F$ at the expense of one less coding gain.  In terms of rate, the new constructions decrease $F$ significantly at the expense of a diminishing loss in  rate as $K$ becomes large.

It should be noted that we only focused on centralized network in this paper. In fact, PDA can also be used to describe decentralized networks where the placement is random.
Accordingly, it requires that the positions of symbol $``*"$ are independent of the users.

\end{document}